\documentclass[runningheads]{llncs}
\usepackage{graphicx}
\usepackage{amsmath}
\usepackage{graphicx}
\usepackage{comment}
\usepackage{todonotes}
\begin{document}

\title{Space Complexity of Streaming Algorithms on Universal Quantum Computers}

\titlerunning{Streaming Algorithms on Universal Quantum Computers}
%
\author{Yanglin Hu \and
Darya Melnyk \and
Yuyi Wang       \and
Roger Wattenhofer}
\authorrunning{Y. Hu et al.}
%
\institute{ETH Zurich, Switzerland
 \\\email{\{yahu,dmelnyk,yuwang,wattenhofer\}@ethz.ch}}
\maketitle              
\begin{abstract}
Universal quantum computers are the only general purpose quantum computers known that can be implemented as of today. These computers consist of a classical memory component which controls the quantum memory. In this paper, the space complexity of some data stream problems, such as PartialMOD and Equality, is investigated on universal quantum computers. The quantum algorithms for these problems are believed to outperform their classical counterparts. Universal quantum computers, however, need classical bits for controlling quantum gates in addition to qubits. Our analysis shows that the number of classical bits used in quantum algorithms is equal to or even larger than that of classical bits used in corresponding classical algorithms. These results suggest that there is no advantage of implementing certain data stream problems on universal quantum computers instead of classical computers when space complexity is considered.

\keywords{Streaming algorithm \and Universal quantum computer \and Space complexity \and Solovay-Kitaev algorithm.}
\end{abstract}

\section{Introduction}
In the past two decades, scientists have made significant progress in the field of quantum computation. Quantum computer protocols based on different physical principles have been constructed and manufactured. Despite this progress, large-scale quantum computers are still not available. 

According to the no-programming Theorem \cite{PhysRevLett.79.321}, a quantum-controlled quantum computer is not better than a classically controlled quantum computer. Therefore, a modern quantum computer consists of a large classical memory controlling a small quantum memory. The limited quantum memory poses great challenges to physicists and computer scientists. In particular, one must decide how to use this limited quantum memory efficiently. One possible way is to build larger-scale quantum computers. Another way is to introduce algorithms that require a small quantum memory, but a large classical memory. In this work, we address the latter case for a special class of problems -- the data stream problems.

Data stream problems process data streams where the input data comes at a high rate. The massive input data challenges communication, computation, and storage. In particular, one may not be able to transmit, compute and store the whole input. For such problems, classical and quantum algorithms have been proposed with the aim to reduce space complexity. On quantum computers, such algorithms usually use polynomially or even exponentially less quantum memory than their classical counterparts using classical memory. 

However, quantum algorithms are generally performed on a universal quantum computer. Note that for some structures quantum gates can change continuously by slowly varying some physical parameters, and it seems that one should use a continuous set of quantum gates to describe them. However, by considering the uncertainty principle, physical parameters can only be measured with errors. Due to these errors, quantum gates with slightly different parameters can therefore often not be distinguished and should be regarded as the same quantum gate. This brings us back to a discrete set of quantum gates and a universal quantum computer. On such a universal quantum computer, only a finite set of quantum gates -- the universal quantum gates -- can be used directly. Other quantum gates are approximated by quantum gate array to a certain accuracy. 

According to the no-programming theorem, universal quantum computers need extra memory, in particular, they need classical memory in order to store the program for the desired quantum gate array. Therefore, the length of the desired quantum gate array would determine the length of the program, which requires extra memory. In this work, we include the extra memory for programs when considering the space complexity, and show that if the extra memory is taken into account, the space complexity of the proposed quantum algorithm for the PartialMOD problem is approximately equal to the space complexity of the respective classical algorithms and that for the Equality problem is even worse. This way, the considered streaming algorithm on universal quantum computers have no advantage over their classical counterparts. Note that our result does not imply that these problems cannot be solved efficiently in a different model. Instead, it suggests that different problems may be solved more efficiently in some particular model, but not in others. We therefore see our result as an inspiration to consider quantum algorithms with respect to the framework in which they can be implemented.

\section{Related Work}

Classical data stream problems have been first formalized and popularized by Alon et al.\ \cite{Alon:1996:SCA:237814.237823} in order to estimate the frequency moment of a sequence using as little memory as possible. The PartialMOD \cite{AMBAINIS2012289} and Equality-like problems \cite{Yao:1979:CQR:800135.804414} are well-known examples of problems in this class. For the PartialMOD problem, Ambainis et al.\ \cite{AMBAINIS2012289} proved a tight bound of $\log p$ bits in the deterministic setting. Ablayev et al. \cite{ABLAYEV2005145,ablayev2008complexity} proved a tight bound of $n$ bits for the deterministic classical streaming algorithms computing Equality problems. 

For the quantum version of data stream problems, Watrous \cite{Watrous:1998:SQC:928244} proved the well-known result that the complexity class PrSPACE(s) is equal to the complexity class PrQSPACE(s), which implies that to some extent, quantum algorithms are not better than classical algorithms with respect to their space complexity. For the PartialMOD problem, Ablayev et al.\ \cite{10.1007/978-3-319-09704-6_6} proposed a quantum algorithm that requires only $1$ qubit while classical algorithms need $\log p$ bits. Ablayev et al.\ \cite{ablayev2008complexity} later proposed quantum streaming algorithms for Equality Boolean functions. Their results show that some problems have both logarithmic or better quantum algorithms, whereas at least a logarithmic number of bits is needed for classical algorithms. Based on both previous results, Khadiev et al. \cite{DBLP:journals/corr/abs-1709-08409,DBLP:journals/corr/abs-1710-09595,DBLP:journals/corr/abs-1802-05134} proposed quantum stream algorithms with constant space complexity, which is better than classical streaming algorithms that require polylogarithmically many bits. Le Gall \cite{LeGall2009} also investigated a certain variation of the Equality problem and proposed a quantum algorithm with exponentially lower space complexity (both quantum and classical) than the corresponding classical algorithm. 

The field of communication complexity also investigated Equality problems. Buhrman et al.\ \cite{PhysRevLett.87.167902} introduced quantum fingerprinting and proposed to use it in communication theory. They chose the Equality problem as an example in their paper. Recently, Guan et al.\ \cite{PhysRevLett.116.240502} managed to realize the above progress experimentally. 

In our paper, we focus on the space complexity of data stream problems on universal quantum computers. For such computers, the Solovay-Kitaev algorithm \cite{Dawson:2006:SA:2011679.2011685,Kitaev:2002:CQC:863293,doi:10.1063/1.1495899} states that any operator can be approximated to an accuracy of $\epsilon$ by $\log^c \frac {1}{\epsilon}$ quantum gates from a finite set of gates. Different versions of the Solovay-Kitaev algorithm consider different values of $c$. In \cite{Dawson:2006:SA:2011679.2011685}, Dawson and Nielsen introduced a version with $c\approx 4$. Kitaev et al. \cite{Kitaev:2002:CQC:863293} proposed a version with $c\approx 2$, and Harrow et al. \cite{doi:10.1063/1.1495899} finally proved a lower bound of $c=1$. Moreover, they showed that the corresponding algorithm exists but cannot be given explicitly.

\section{Background}
We will start by describing the notation for quantum computation we use in this paper. We will introduce fundamental concepts of quantum physics using the Dirac notation, and also present the Bloch sphere model, which is a geometric way to comprehend quantum algorithms. In Section \ref{sec:solovay-kitaev}, we will then clarify the Solovay-Kitaev algorithm \cite{Dawson:2006:SA:2011679.2011685} which gives a way to efficiently approximate any desired operation on a universal quantum computer with a finite set of operations. Finally, in Section \ref{sec:no-programming}, we will explain the quantum no-programming theorem \cite{PhysRevLett.79.321}. This theorem points out that we must use orthonormal quantum states to perform different operators with deterministic quantum gate arrays and as such it forms the basis of a classically controlled quantum computer.  

\subsection{Notation}\label{sec:notation}
Let $|i\rangle $ denote the $i$-th classical state of the (complete orthonormal) computational basis of a Hilbert space. We can write a pure state of a quantum memory as a column vector $|\psi\rangle = \left(\alpha_1,...,\alpha_n\right)^T =\sum_{i=1}^{n}{\alpha_i|i\rangle}$. Its norm satisfies $\langle \psi|\psi\rangle = \sum_{i=1}^{n}{|\alpha_i|^2}=1$, 
where $\langle\psi|$ is the conjugate transpose of $|\psi\rangle$. 

The evolution of a state can be represented by a unitary operator $U$. 
That is, given an initial state  $|\psi\rangle$, the final state after applying $U$ is  
$|\psi'\rangle=U|\psi\rangle$. It is easy to verify that the norm of a state does not change after an evolution. 

A single-qubit memory can be represented as a point on the so-called Bloch sphere. Explicitly, a unitary operator 
$$ U=
\begin{pmatrix}
    \cos(\frac {\theta}{2}) &   -e^{-i\phi} \sin(\frac {\theta}{2}) \\
    e^{i\phi}\sin(\frac {\theta}{2}) & \cos (\frac {\theta}{2}) 
\end{pmatrix}. 
$$

\noindent corresponds to the vector $U|0\rangle$, pointing to $(\theta,\phi)$ on the sphere, where $\theta$ is the angle between the vector and the $z$-axis, and $\phi$ is the angle between the projection of the vector onto $xOy$ plane and the $x$-axis.
Note that the north pole corresponds to all $(0,\ \phi)$ and the south pole corresponds to all $(\pi,\phi)$. 

A projective measurement can be represented by a set of orthogonal projectors followed by a normalization. That is, if we apply a measurement $\{P_i\}$ to the initial state  $|\psi\rangle$, the final state becomes $|\psi_i'\rangle = P_i|\psi\rangle / \|P_i|\psi\rangle\|$ with probability $|\langle \psi|P|\psi\rangle|$. In particular, if the measurement is $\{|i\rangle\langle i|\}$ and the initial state is $\sum {\alpha_i|i\rangle}$, the final state is $|i\rangle$ with  probability $|\alpha_i|^2$. Note that the total probability of all possible  final states is $1$. 

\subsection{Solovay-Kitaev Algorithm}\label{sec:solovay-kitaev}

In order to present the Solovay-Kitaev algorithm, we first need to introduce the concept of universality.

\begin{definition}[Universal quantum gates \cite{nielsen_chuang_2010}]
A set of quantum gates is universal for quantum computation if any unitary operator can be approximated to arbitrary accuracy by a quantum circuit involving only these gates.
\end{definition}

Note that an example of such a set can be found in Chapter 4 of \cite{nielsen_chuang_2010}. 

Based on the well-defined universal set, we can now state the Solovay-Kitaev theorem which talks about how efficient a universal set is:

\begin{theorem}[Solovay-Kitaev \cite{Dawson:2006:SA:2011679.2011685}]\label{thm:Solovay-Kitaev}
There exist algorithms that can approximate any unitary operator $U$ to an accuracy of $\|U-U_{approx}\|_2\leq\epsilon$ with $O(\log^c \frac 1 {\epsilon})$ universal quantum gates. 
\end{theorem}

The proof of the theorem can be found in \cite{Dawson:2006:SA:2011679.2011685}. 

According to Harrow \cite{doi:10.1063/1.1495899}, $\Omega(\log \frac 1 {\epsilon})$ quantum gates are needed in order to approximate any unitary operator in two dimensions to an accuracy of $\epsilon$. The Solovay-Kitaev algorithm is optimal if we disregard poly-logarithmic differences in the number of quantum gates. 

The Solovay-Kitaev theorem does not exclude the possibility that we can approximate some unitary operator with a quantum gate array much shorter than $O(\log \frac 1 {\epsilon})$ to an accuracy of $\epsilon$.  

\subsection{No-Programming Theorem}\label{sec:no-programming}

The no-programming theorem \cite{PhysRevLett.79.321} shows that we cannot use fewer qubits than classical bits for programming if we want to implement a quantum gate array deterministically. We view our quantum computer as a unitary operator $G$ acting on both the quantum program $|P\rangle$ and the memory $|d\rangle$. $G$ acting on a quantum program $|P\rangle$ for unitary $U$ results in $U|d\rangle \otimes |P'\rangle$. After measurement we get $U|d\rangle$ \textit{deterministically}. However, $G$ acting on a superposition of orthogonal quantum programs $\frac {1} {\sqrt 2} (|P_1\rangle+|P_2\rangle)$ results in a superposition of orthogonal states $\frac {1} {\sqrt 2} (U_1|d\rangle \otimes |P'_1\rangle+U_2|d\rangle\otimes |P'_2\rangle)$. Therefore, after our measurement, we obtain either $U_1|d\rangle$ or $U_2|d\rangle$ \textit{stochastically}.  

\begin{theorem}\label{thm:no_programming}
On a fixed, general purposed quantum computer, if we want to deterministically implement a quantum gate array, quantum programs $|P_1\rangle$,...,$|P_n\rangle$ performing distinct unitary operator $U_1$,...,$U_n$ are orthogonal. The program memory is at least $N$-dimensional, that is, it contains at least $\log(N)$ qubits. 
\end{theorem}

The theorem shows that, when used for programming a deterministic quantum gate array, a quantum program has no advantage over a classical program, i.e. in this aspect a quantum controlled quantum computer is no better than a classically controlled quantum computer. 

When used for programming a probabilistic quantum gate array, there are quantum programs that use exponentially less space but succeed with exponentially smaller probability, which is not practical. In our paper, we thus only consider classical bits for programming.

\section{Data Stream Problems}
In this section, we present selected examples of data stream problems and study their space complexity. Each section is organized as follows: we first introduce the problem statement and the corresponding proposed algorithm for quantum computers with a continuous set of gates. In practice, quantum computers with a continuous set of gates cannot be realized, which makes such algorithms only of theoretic interest. In the following section, we assume that our universal quantum computer first selects a certain universal set of gates, then it is asked data stream problems with any possible scale and parameter. The quantum computer should answer any possible question using the same universal set of gates. We therefore analyze the space complexity of the respective algorithm on such a universal quantum computer and show that it has no advantage over the space complexity of the best known classical algorithm. 

\subsection{PartialMOD Problem}
In this section, we study the PartialMOD problem as presented in \cite{DBLP:journals/corr/abs-1802-05134,AMBAINIS2012289,ABLAYEV2005145}. 
In this problem, we receive some unknown bitstring bit by bit of which we know that the number of bits with value $1$ is a multiple of a given number. The task is to determine the parity of the multiplier of this number while storing as few bits as possible in the memory.

\begin{definition}[PartialMOD problem]
Let $(x_1,...,x_n)$ be an input sequence of classical bits. Assume that we know in advance that $\#_1$ is a multiple of $p$, i.e., $\#_1=v\cdot p$, where $\#_1$ denotes the number of ones in the string. The bits are received one by one by the algorithm. The problem is to determine the parity of $v$, i.e., to output $v\ mod\ 2$. 
\end{definition}

\subsubsection{Algorithm with a Continuous Set of Gates} 

Ambainis and Yakaryilmaz \cite{AMBAINIS2012289} showed that there exists no deterministic or probabilistic algorithm to compute PartialMOD problem with $o(\log p)$ classical bits. In their paper, they also propose a quantum algorithm solving PartialMOD using only one qubit.
This algorithm works as follows: There is only one qubit in the quantum memory. Let the initial state of the qubit be $|0\rangle$, which is the north pole of the Bloch sphere, and set $\theta_p =\frac {\pi}{2p}$. Each time we receive a $1$ as the next bit, we apply a unitary operator
$$R(\theta_p)= 
\begin{pmatrix}
        \cos \theta_p   &   \sin \theta_p\\
        -\sin \theta_p  &   \cos \theta_p
\end{pmatrix}.
$$
on the qubit, which is a rotation by $2\theta_p$ around $y$-axis on the Bloch sphere. After $v\cdot p$ steps, we receive all the input bit and get the state 
\begin{equation*}
    |\psi_f\rangle = 
    \left( \cos (v\frac {\pi} {2}) \ \ -\sin (v\frac {\pi} {2})\right)^\text{T}.
\end{equation*}

If $v\ mod\ 2=0$, we return to the north pole of the Bloch sphere and the final state of the qubit is $|0\rangle$. If $v\ mod\ 2=1$, we reach the south pole of the Bloch sphere and the final state is $|1\rangle$. Finally, we can measure the qubit and obtain its state. 

With this procedure, we only need one qubit to solve the PartialMOD problem on quantum computers. In contrast, a classical computer requires to use $\log p$ bits, as is shown in \cite{ABLAYEV2005145}.

\subsubsection{Analysis on Universal Quantum Computers}

In the following, we show that the proposed quantum algorithm is not space efficient on universal quantum computers. We suppose that our universal quantum computer is able to solve any specific PartialMOD problem, which requires that we should be able to apply any $R(\theta_p)$ to the demanded accuracy. Observe that there are infinitely many choices of $p$, and thus infinitely many different $R(\theta_p)$. Since only finitely many gates can be selected in the universal set of a quantum computer, $R(\theta_p)$ have to be approximated by a quantum gate array, where each gate of the array is from the universal set. This leads to possibly wrong outputs. Assume therefore that we approximate $R(\theta_p)$ by $R(\theta_p+\epsilon_p)$, which satisfies
\begin{equation*}
    \|R(\theta_p)-R(\theta_p+\epsilon_p)\|_2=4\sin {\frac {\epsilon_p}{4}}.
\end{equation*}
Starting with the initial state $|0\rangle$, we reach the state
\begin{equation*}
    |\psi_f\rangle = 
    \begin{pmatrix}   
        \cos (v\frac {\pi} {2}+v p \cdot\epsilon_p)\\
        -\sin (v\frac {\pi} {2}+v p \cdot\epsilon_p)
    \end{pmatrix}.
\end{equation*} 
after $vp$ steps. With probability $\sin^2 (v p \cdot\epsilon_p)$ we may get an incorrect output from the measurement. If $v p \cdot\epsilon_p$ is small enough, we can bound the probability of an incorrect output by a positive constant $\delta$ as follows 
\begin{equation*}
    \frac{1}{2}v p\cdot 4\sin \left(\frac {\epsilon_p}{4}\right)\leq \sin\left(v p \cdot\epsilon_p\right)\leq \sqrt {\delta}.
\end{equation*} 
Therefore, an accuracy of $\frac {2\sqrt{\delta}}{vp}$ must be achieved. Such an accuracy comes at the cost of additional quantum gates. 

Intuitively, applying the Solovay-Kitaev algorithm, we need a quantum gate array of $\log\frac{vp}{2\sqrt{\delta}}$. We will show next that a quantum gate array of at least $\Omega(\log ( \frac{v} {\sqrt{\delta}} \log p))$ gates must be used in order to approximate $R(\theta_p)$ in the proposed algorithm to an accuracy of $\frac {2\sqrt{\delta}}{vp}$. Note that in this theorem we do not assume the optimality of the Solovay-Kitaev algorithm. Because the optimality of Solovay-Kitaev algorithm is in the sense of polylogarithmic equivalence, and the truly optimal algorithm has not been given, simply assuming this algorithm to be optimal may cause difficulties. However, even without such an assumption, Theorem \ref{thm:approximation} still shows that the quantum algorithm performs worse in some situations. 

\begin{theorem}\label{thm:approximation}
No algorithm can approximate all $R(\theta_p)$, where $\theta_p=\frac{\epsilon}{2p}$ and $p\leq p_0$, to an accuracy of $\epsilon_p=\frac {\epsilon}{2p}$ using $o(\log(\frac {1} {\epsilon}\log{p_0}))$ quantum gates on a universal computer, where $p_0$ is sufficiently large and $\epsilon$ sufficiently small. We do not assume the optimality of the Solovay-Kitaev algorithm here. 
\end{theorem}

We will not present the proof here, but the general idea of the proof is inspired by \cite{doi:10.1063/1.1495899}.

Theorem \ref{thm:approximation} implies that at least $O(\log (\frac {v}{\sqrt{\delta}}\log p_0))$ quantum gates are needed in order to approximate all $R(\theta_p)$, where $p\leq p_0$, to the demanded accuracy of $\epsilon_p = \frac {2\sqrt{\delta}}{vp}$ in order to ensure a success probability of at least $1-\delta$. Since we have to store the arrangement of the quantum gate array for each $R(\theta_p)$, the number of classical bits required is equal to the number of gates in the quantum gate array, that is, at least $O(\log (\frac {v}{\sqrt{\delta}}\log p))$ classical bits. It is obvious that when $v$ approaches infinity while $p$ remains finite, the quantum algorithm for PartialMOD is not more space-efficient than the corresponding classical algorithm.

Assuming the optimality of the Solovay-Kitaev algorithm, which is discussed in \cite{doi:10.1063/1.1495899}, we can also show that in order to obtain such an accuracy, at least $O(\log (\frac {vp} {\sqrt{\delta} }))$ quantum gates must be used by any algorithm.  

\begin{theorem}\label{thm:another_approximation}
Let $p$ be sufficiently large. No algorithm can approximate all $R(\theta_p)$, where $\theta_p=\frac {\pi}{2p} $ to any accuracy $\epsilon_p$ with $o(\log(\frac 1 {\epsilon_p}))$ quantum gates on a universal computer, if the optimality of the Solovay-Kitaev algorithm is assumed. 
\end{theorem}

This theorem can be proved by contradiction: if one can approximate these operators with $o(\log(\frac 1 {\epsilon_p}))$ quantum gates, then it is possible to construct a better algorithm than the Solovay-Kitaev algorithm.  

Theorem \ref{thm:another_approximation} shows that there exists some $R(\theta_p)$ for which we need at least $\Omega(\log \frac 1 {\epsilon_p})=\Omega(\log v+\log p)$ gates in order to approximate it to the demanded accuracy of $\epsilon_p={\frac {2\sqrt{\delta}} {vp}}$, assuming the optimality of the Solovay-Kitaev theorem. Since we have to store the arrangement of the quantum gate array, the number of classical bits needed is $\Omega(\log v+\log p)$. When $v$ or $p$ approach infinity, the quantum algorithm is not more space-efficient than the classical algorithm. 

Theorem \ref{thm:approximation} and Theorem \ref{thm:another_approximation} are proved under different assumptions. Together they show that the previously proposed algorithm is not more space-efficient than its classical counterpart under certain conditions. 

\subsection{Equality Problem}
In this section, we investigate the so-called Equality problem \cite{ablayev2008complexity,Newman:1996:PVP:237814.238004,PhysRevLett.87.167902}.
In this problem, two bitstrings are received once one after another bit by bit. The task is to find out whether these two given sequences of bits are equal while storing a minimal amount of information. 

\begin{definition}[Equality problem]\label{def:equality_problem}
We are given an input sequence $(x,y)=(x_1,...,x_n,y_1,...,y_n)$ of classical bits. The bits are received one by one by the algorithm. We do not receive any bit of $y$ before we have received all bits of $x$. The output is whether $x$ and $y$ are equal, i.e., $O=\delta(\|x-y\|)=1,x=y;0,x\neq y$. 
\end{definition}

\subsubsection{Algorithm with a Continuous Set of Gates}\label{sec:EQonNONuni}
According to \cite{SAUERHOFF2005177,Newman:1996:PVP:237814.238004}, there is no classical deterministic algorithm that can compute the equality problem with $o(n)$ classical bits, while there is a randomized algorithm, i.e. Karp-Rabin algorithm, with a space complexity of $O(\log n)$ \cite{karp1987efficient}. There also exists a quantum algorithm that has the same performance. Ablayev et al.~\cite{ablayev2008complexity} applied quantum fingerprinting in a quantum streaming algorithm to solve this problem with $O(\log n)$ qubits on a quantum computer with a continuous set of gates. Their algorithm seems to have the same performance as the Karp-Rabin algorithm.  

The quantum memory is divided into two parts. The first part is the first qubit, whose state is in a 2-dimensional space. The second part contains the remaining $\log t$ qubits in a $t$-dimensional space. The initial state is $|0\rangle \otimes |0\rangle $. The strategy is to first apply Hadamard gates on all qubits of the second part and receive $\frac 1 {\sqrt {t}}|0\rangle \otimes \sum_{j=1}^{t}|j\rangle $. If we receive a $1$ for $x_i$, we apply a unitary operator $U_i=\sum_{j=1}^{t} \{R(\theta_{ij})\otimes |j\rangle \langle j|\}$, where $R(\theta_{ij})$ is a rotation on the first qubit by $\theta_{ij}=\frac {2\pi m_j}{2^{i+1}}$ and $m_j$ some positive integer.
If we receive a $1$ for $y_i$, we replace $R(\theta_{ij})$ with $R(-\theta_{ij})$ in $U_i$. After receiving all the input bits, the state is 
\begin{equation*}
    \frac 1 {\sqrt{t}}\sum_{j}{  R\left(\frac {2\pi m_j(x-y)}{2^{n+1}}\right)|0\rangle\otimes|j\rangle}.
\end{equation*}
Then we apply Hadamard gates on all qubits in the second part. The final state becomes
\begin{equation*}
    \frac 1 t \sum_{j}{\left(\cos \frac {2\pi m_j(x-y)}{2^{n+1}} \right)}|0\rangle \otimes |0\rangle+ rest.
\end{equation*}
If $x=y$, we return to the initial state. If $x\neq y$, we reach a non-initial state. We require that the coefficient of $|0\rangle\otimes |0\rangle$ in the final state is approximately a delta function, that is,
\begin{equation*}
    \left\|\frac 1 t \sum_{j}{\cos \frac {2\pi m_j(x-y)}{2^{n+1}}}-\delta(x-y) \right\|\leq \sqrt{\epsilon}.
\end{equation*}
Then we can easily verify whether $x=y$ by checking whether we get $|0\rangle\otimes |0\rangle $ after measurement. If $x=y$, we obtain $|0\rangle$ with probability $1$. If $x\neq y$, we obtain $|0\rangle \otimes |0\rangle$ with probability less than $\epsilon$.

If we apply discrete Fourier transform to $\delta(g)$, that is, $m_j$s take $t=2^n$ integers from $0$ to $2^n$, $\epsilon$ is exactly $0$. But in that case we need $\log(t)=n$ qubits. It is however possible that if we do not apply discrete Fourier transform, that is, $m_j$s only take $t=O(n \log \frac 1 {\epsilon})\ll 2^n$ integers from $0$ to $2^n$, $\epsilon$ is also bounded. The next theorem states this fact, its proof can be found in \cite{ablayev2008complexity}. 

\begin{theorem}\label{thm:funny_theorem}
There exists a set of $t>\frac 2 {\epsilon} \ln(2m)$ elements, $\{m_j,\ j=1,...,t\}$ such that 
\begin{equation*}
    \frac 1 {t} \left\|\sum_{j} \cos \left(\frac {2\pi m_j g} {m}\right)\right\| \leq \sqrt{\epsilon},\ \forall g\neq 0.
\end{equation*}
\end{theorem}

Theorem \ref{thm:funny_theorem} implies that there exists a set of $t=\frac 2 {\epsilon} \ln (2m) +1$ elements, $\{m_j,j=1,...,t\}$, which ensures $\cos\left(\frac {2\pi k_i g}{m}\right)$'s to almost cancel each other. Indeed, if we select integers uniformly at random from $0$ to $m-1$, we are likely to get such $m_j$.
By applying Theorem \ref{thm:funny_theorem} to the Equality problem, we only need $\log(n)+1$ qubits on quantum computers with a continuous set of gates, which is exponentially better than $n$ bits deterministic algorithms on computers, as was shown by Babai et al.~\cite{Babai:1997:RSM:791230.792293}.

\subsubsection{Analysis on Universal Quantum Computers}
The proposed algorithm to solve the Equality problem is not space-efficient on universal quantum computers. Similar to the PartialMOD problem, we will first bound the accuracy of each operator. Let us denote the probability for the algorithm to accept the input, i.e., in the case where the final state is $|0\rangle \otimes |0\rangle$, as $\Pr(x,y)$. Further, assume that it is possible to approximate the operator $R(\theta_{ij})$ to an accuracy of $\delta_{ij}$. After applying Theorem \ref{thm:funny_theorem}, the partial derivative of $\Pr(x,y)$ becomes 
\begin{equation*}
    \delta \Pr(x,y)\leq\frac{\sqrt{\epsilon}}{t}\left\|\sum_{j=1}^t\sum_{i}\sin\left(\frac{m_j\pi 2^i(x-y)_i}{2^n}\right)\delta_{ij}\right\|.
\end{equation*}
Differently than in the PartialMOD problem, it is challenging to bound the accuracy for the Equality problem precisely. Instead, we simply assume we need $\Omega(1)$ gates for each $R(\theta_{ij})$. The following theorem defines an upper bound on the accuracy needed, and shows our simple assumption is reasonable. 

\begin{theorem}\label{thm:EQ_Upper_Bound}
Let $|\delta_{ij}|\leq \frac 1 n$. Then, there exists a set of $t=\frac 2 {\epsilon}(n+3)$ elements $m_j,j=1,...,t$, such that the following two inequalities are satisfied
\begin{equation*}
    \frac {1} {t} \left\|\sum_{j} \cos \left(\frac {\pi m_j g} {2^n}\right)\right\| \leq \sqrt{\epsilon},\ \forall g\neq 0,
\end{equation*}
\noindent and
\begin{equation*}
    \frac {1} {t}\left\|\sum_{j=1}^t\sum_{i}\sin\left(\frac{m_j\pi 2^ig_i}{2^n}\right)\delta_{ij}\right\|\leq \sqrt{\epsilon},\ \forall g\neq 0.
\end{equation*}
\end{theorem}

Here, we prove it via a method similar to that of Theorem \ref{thm:funny_theorem}, shown in \cite{ablayev2008complexity}.

The algorithm for the Equality problem will succeed as long as we reach an accuracy of $\frac 1 n$ for a suitably chosen set of $O(n)$ elements. In order to achieve such accuracy, we need at most $O(\log^4 n)$ quantum gates according to the Solovay-Kitaev theorem. Since we need to apply at least one quantum gate in order to be able to implement an operator, it is reasonable to assume that we need at least $\Omega(1)$ quantum gates for each operator to achieve such accuracy. 

Now we can analyze the space complexity, for which we also take into account classical bits. When we perform the above algorithm we need to store the set $\{m_j\}$, since the set $\{m_j\}$ is not chosen arbitrarily. There are two natural ways to do so. One way is to store $\{m_j\}$ directly: consider $m_j$ that range from $0$ to $2^n$, and thus need $n$ classical bits. We have $n$ such integers, and thus at least $\Omega(n^2)$ bits are needed. This strategy requires even more bits than a classical brute force method which saves all $O(n)$ bits of the input. The second way is to store $\{R(\theta_{ij})\}$: note that $R(\theta_{ij})$ need at least $\Omega(1)$ quantum gates for each operator, and thus each need $\Omega(1)$ classical bits. Since we have $n^2$ such operators in our algorithm, at least $\Omega(n^2)$ bits of storage are needed, which is more than that in the classical deterministic algorithm. In the following theorem, we provide a more rigorous proof. 

\begin{theorem}
At least $\Omega(n^2)$ bits are needed in order to store a set $\left\{m_j, j=1...t \right\}$ where $m_j\in [0,2^n-1]$ and $t=\frac{2}{\epsilon}(n+3)$ without pre-knowledge of the set.
\end{theorem}
\begin{proof}
We first consider the classical case. The number of possible choices in the classical case is 
\begin{equation*}
    C_{t}^{2^n} = \frac{2^n!}{t!\cdot (2^n-t)!}.
\end{equation*}
The information entropy of knowing a certain choice from all possible choices with equal possibility is
$ S=\ln\left(C_{t}^{2^n}\right).$
Consider when $n$ is sufficiently large, $2^n \gg t=\frac{2}{\epsilon}(n+3)$, use $\ln(1+x)\approx x$ an $\ln(x!)\approx x\ln(x)-x$, we have
\begin{equation*}
    S=\ln(2^n)+...+\ln(2^n-t+1)-\ln(t!)\approx nt-\frac{t^2}{2^n}-t\ln(t)+t=O(n^2).
\end{equation*}
Since the number of bits required is linearly dependent on the information entropy, $O(n^2)$ bits are needed in order to store this set. 
    
We next consider the quantum case. The set $\left\{m_j, j=1...t,m_j \in [0,2^n-1]  \right\}$ is used to program our quantum computer. Due to the quantum no-programming theorem in \ref{sec:no-programming}, quantum programs have no advantage over the classical program with respect to space complexity. Therefore, $\Omega(n^2)$ bits or qubits are needed to store this set.
\end{proof}

Therefore, the considered algorithm for the Equality problem has no advantage over the classical deterministic algorithm.

\section{Conclusion}

Based on the Solovay-Kitaev algorithm, we investigated the space complexity of streaming algorithms on a universal computer when only a finite number of quantum gates are available. We used the PartialMOD problem and the Equality problem to analyze the quantum streaming algorithms in systems where classical bits are used in order to control quantum gates. By applying the Solovay-Kitaev algorithm we concluded that the considered quantum streaming algorithms do not beat their classical counterparts in this system. 

Our work shows that not all quantum streaming algorithms can perform well on a universal quantum computer. There are also data stream problems for which quantum algorithms may perform well on a universal quantum computer. One example is the variation of the Equality problem proposed in \cite{LeGall2009}. In this problem, the input is repeated many times, which is different from the Equality problem discussed in this paper, where we receive the input only once. Another possible candidate is the problem based on the universal $(\epsilon,l,m)$-code of matrices proposed by Sauerhoff et al. in \cite{SAUERHOFF2005177} and Gavinsky et al. \cite{Gavinsky:2007:ESO:1250790.1250866}, where the input directly corresponds to a quantum gate array, and one can therefore save space when storing quantum gates for application. By comparing these algorithms, we conclude that a framework can be extremely efficient for a certain set of problems and corresponding algorithms, but not necessarily for all problems. We therefore think that the space complexity of algorithms should be analyzed with respect to the framework of the quantum computer in which they can be implemented.

\bibliography{literature}

\begin{thebibliography}{10}
\providecommand{\url}[1]{\texttt{#1}}
\providecommand{\urlprefix}{URL }
\providecommand{\doi}[1]{https://doi.org/#1}

\bibitem{ABLAYEV2005145}
Ablayev, F., Gainutdinova, A., Karpinski, M., Moore, C., Pollett, C.: {On the
  computational power of probabilistic and quantum branching program}.
  {Information and Computation}  \textbf{203}(2) (2005)

\bibitem{10.1007/978-3-319-09704-6_6}
Ablayev, F., Gainutdinova, A., Khadiev, K., Yakary{\i}lmaz, A.: {Very Narrow
  Quantum OBDDs and Width Hierarchies for Classical OBDDs}. In: Descriptional
  Complexity of Formal Systems. Springer International Publishing (2014)

\bibitem{ablayev2008complexity}
Ablayev, F., Khasianov, A., Vasiliev, A.: {On complexity of quantum branching
  programs computing equality-like boolean functions}. {Electronic Colloquium
  on Computational Complexity}  (2010)

\bibitem{Alon:1996:SCA:237814.237823}
Alon, N., Matias, Y., Szegedy, M.: {The Space Complexity of Approximating the
  Frequency Moments}. In: {Proceedings of the Twenty-eighth Annual ACM
  Symposium on Theory of Computing}. STOC (1996)

\bibitem{AMBAINIS2012289}
Ambainis, A., Yakaryılmaz, A.: {Superiority of exact quantum automata for
  promise problems}. {Information Processing Letters}  \textbf{112}(7) (2012)

\bibitem{azuma1967}
Azuma, K.: Weighted sums of certain dependent random variables. {Tohoku
  Mathematical Journal}  \textbf{19}(3),  357--367 (1967)

\bibitem{Babai:1997:RSM:791230.792293}
Babai, L., Kimmel, P.G.: Randomized simultaneous messages: Solution of a
  problem of yao in communication complexity. In: Proceedings of the 12th
  Annual IEEE Conference on Computational Complexity. CCC (1997)

\bibitem{814621}
Boykin, P.O., Mor, T., Pulver, M., Roychowdhury, V., Vatan, F.: {On universal
  and fault-tolerant quantum computing: a novel basis and a new constructive
  proof of universality for Shor's basis}. In: {40th Annual Symposium on
  Foundations of Computer Science} (1999)

\bibitem{PhysRevLett.87.167902}
Buhrman, H., Cleve, R., Watrous, J., de~Wolf, R.: {Quantum Fingerprinting}.
  {Physical Review Letters}  \textbf{87} (2001)

\bibitem{Dawson:2006:SA:2011679.2011685}
Dawson, C.M., Nielsen, M.A.: {The Solovay-Kitaev Algorithm}. {Quantum
  Information and Computation}  \textbf{6}(1) (2006)

\bibitem{Gavinsky:2007:ESO:1250790.1250866}
Gavinsky, D., Kempe, J., Kerenidis, I., Raz, R., de~Wolf, R.: {Exponential
  Separations for One-way Quantum Communication Complexity, with Applications
  to Cryptography}. In: {Proceedings of the Thirty-ninth Annual ACM Symposium
  on Theory of Computing}. STOC (2007)

\bibitem{PhysRevLett.116.240502}
Guan, J.Y., Xu, F., Yin, H.L., Li, Y., Zhang, W.J., Chen, S.J., Yang, X.Y., Li,
  L., You, L.X., Chen, T.Y., Wang, Z., Zhang, Q., Pan, J.W.: {Observation of
  Quantum Fingerprinting Beating the Classical Limit}. {Physical Review
  Letters}  \textbf{116} (2016)

\bibitem{doi:10.1063/1.1495899}
Harrow, A.W., Recht, B., Chuang, I.L.: {Efficient discrete approximations of
  quantum gates}. {Journal of Mathematical Physics}  \textbf{43}(9) (2002)

\bibitem{karp1987efficient}
Karp, R.M., Rabin, M.O.: Efficient randomized pattern-matching algorithms. IBM
  journal of research and development  \textbf{31}(2),  249--260 (1987)

\bibitem{DBLP:journals/corr/abs-1710-09595}
Khadiev, K., Khadieva, A., Kravchenko, D., Rivosh, A.: {Quantum versus
  Classical Online Algorithms with Advice and Logarithmic Space}  (2017)

\bibitem{DBLP:journals/corr/abs-1709-08409}
Khadiev, K., Khadieva, A., Mannapov, I.: Quantum online algorithms with respect
  to space complexity. Lobachevskii Journal of Mathematics  \textbf{39} (2017)

\bibitem{DBLP:journals/corr/abs-1802-05134}
Khadiev, K., Ziatdinov, M., Mannapov, I., Khadieva, A., Yamilov, R.: {Quantum
  Online Streaming Algorithms with Constant Number of Advice Bits}  (2018)

\bibitem{Kitaev:2002:CQC:863293}
Kitaev, A.Y., Shen, A., Vyalyi, M.N.: Classical and Quantum Computation.
  American Mathematical Society, Boston, MA, USA (2002)

\bibitem{LeGall2009}
Le~Gall, F.: Exponential separation of quantum and classical online space
  complexity. Theory of Computing Systems  \textbf{45} (2009)

\bibitem{Newman:1996:PVP:237814.238004}
Newman, I., Szegedy, M.: {Public vs. Private Coin Flips in One Round
  Communication Games (Extended Abstract)}. In: {Proceedings of the
  Twenty-eighth Annual ACM Symposium on Theory of Computing}. STOC (1996)

\bibitem{PhysRevLett.79.321}
Nielsen, M.A., Chuang, I.L.: {Programmable Quantum Gate Arrays}. {Physical
  Review Letters}  \textbf{79} (1997)

\bibitem{nielsen_chuang_2010}
Nielsen, M.A., Chuang, I.L.: {Quantum Computation and Quantum Information: 10th
  Anniversary Edition}. Cambridge University Press (2010)

\bibitem{SAUERHOFF2005177}
Sauerhoff, M., Sieling, D.: {Quantum branching programs and space-bounded
  nonuniform quantum complexity}. {Theoretical Computer Science}
  \textbf{334}(1) (2005)

\bibitem{Watrous:1998:SQC:928244}
Watrous, J.H.: {Space-bounded Quantum Computation}. Ph.D. thesis, The
  University of Wisconsin - Madison (1998)

\bibitem{Yao:1979:CQR:800135.804414}
Yao, A.C.C.: {Some Complexity Questions Related to Distributive Computing
  (Preliminary Report)}. In: {Proceedings of the Eleventh Annual ACM Symposium
  on Theory of Computing}. STOC (1979)

\end{thebibliography}

\newpage

\appendix

\renewcommand{\theequation}{\thesection.\arabic{equation}}
\section{Universality of the Standard Set}\label{app:Univ_Stand_Set}

The standard set consists of Hadamard, phase, controlled-NOT and $\pi /8$ gates. Here, the phase gate is actually the square of the $\pi /8$ gate. However, they are both included in the standard set for fault-tolerant reasons. 

Following the proofs by Nielsen and Chuang \cite{nielsen_chuang_2010}, here we briefly explain why the standard set is universal. First, we can use controlled-NOT gates to entangle two qubits. Combined with all single-qubit gates, any two-qubit gates can be realized. Multi-qubit gates can be realized by composition of two-qubit gates applying on any pair of qubits. Now consider therefore single-qubit operators. Using the Hadamard and $\pi/8$ gate, rotations around the $x$-, $y$- and $z$-axis by $\pi/4$ and $\pi/2$ on the Bloch sphere can be constructed. A rotation by $\pi/4$ around the $z$-axis followed by a rotation by $\pi/4$ around the $x$-axis is equivalent to the rotation by $\lambda \pi$ around the $a$-axis, where $a=(\cos \frac {\pi} {8},\sin \frac {\pi} {8},\cos \frac {\pi} {8})$ and $\lambda$ is an irrational number as proven by Boykin et al. \cite{814621}. If we repeat these rotations, we can approximate any rotation around $a$, since $\lambda$ is irrational. We can then use rotations by $\pi/2$ around the $x$-, $y$- and $z$-axes and by any angle around the $a$ axis in order to approximate any possible rotation. 

\section{Proof of the Solovay-Kitaev Theorem}\label{app:Solovay_Algo_Proof}
In the following, we present a brief summary of the Solovay-Kitaev algorithm. Let $A$ and $B$ be two unitary operators close to the identity, and $\Delta A_0$ and $\Delta B_0$ be the deviations of $A$ and $B$ from the identity, where $\Delta$ denotes the order of the deviation. 
\begin{equation*}
    A=I+\Delta A_0;\ B=I+\Delta B_0
\end{equation*}
Observe that for their commutator holds 
\begin{equation}\label{equation:solovay-kitaev_1}
    ABA^\dagger B^\dagger =I+O(\Delta^2)+O(\Delta^3)
\end{equation}
Suppose now that we can approximate any quantum gate to a basic accuracy of $\epsilon_0$. Let $L_0$ be the length of the quantum gate array that we need to achieve this accuracy. We also have a recursion $U_n=REC(U,n)$ to approximate any operator $U$ to an accuracy of $\epsilon_n$ within $n$ steps. Let $L_n$ be the length of the corresponding quantum gate array. We need to find $V$ and $W$ such that 
\begin{equation*}
    UU_n^\dagger=I+O(\epsilon_n) = VWV^\dagger W^\dagger+O(\Delta_n^3)
\end{equation*}
where $\Delta_n=O(\sqrt{\epsilon_n})$. 
\begin{equation}\label{equation:solovay-kitaev_2}
    VWV^\dagger W^\dagger=I+O(\Delta_n^2)-O(\Delta_n^3)
\end{equation}

Observe by comparing (\ref{equation:solovay-kitaev_1}) and (\ref{equation:solovay-kitaev_2}), that the deviation of $V$ and $W$ from the identity is $O(\Delta_n)$. For a detailed construction of $V$ and $W$ readers may refer to Dawson et al.\ \cite{Dawson:2006:SA:2011679.2011685}. In general, $V$ and $W$ should also be approximated by quantum gate arrays. We use the $n$-step recursion to approximate $V$ and $W$ to an accuracy of $\epsilon_n$. We therefore replace $V$ and $W$ by $V_n$ and $W_n$. The deviation of $V_n W_n V_n^\dagger W_n^\dagger$ from $VWV^\dagger W^\dagger$ is
\begin{equation*}
    V_n W_n V_n^\dagger W_n^\dagger=VWV^\dagger W^\dagger+O(\Delta_n^3)
\end{equation*}   
Finally we have 
\begin{equation*}
    U=V_n W_n V_n^\dagger W_n^\dagger U_n+O(\epsilon_n^\frac {3} {2}) =U_{n+1}+O(\epsilon_{n+1})
\end{equation*}
Then, after the $n$-th recursion, we have 
\begin{equation*}
    \epsilon_n=O(\epsilon_0^{{\frac 3 2}^n});\ L_n=5^n O(L_0)
\end{equation*}
which implies 
\begin{equation*}
    L_n=C\cdot \ln^c\left(\frac {1}{\epsilon_n}\right)
\end{equation*}
where $c=\frac {\log 5} {\log 3/2}\approx 4$. 

Then, we can approximate any rotation on a single qubit to an accuracy of $\epsilon$ using $O(\log^c \frac 1 \epsilon )$ quantum gates, where $c$ varies from $1$ to $4$ depending on the structure of the chosen algorithm. 

There is also a physical image that explains the logarithmic relation between $L_n$ and $\epsilon_n$. For each operator that can be applied exactly, there is a corresponding point on the Bloch sphere representing it. We can use this operator to approximate all operators within its $\epsilon_n$ radius. Equivalently, we can use a point to approximate all the points within its radius $\epsilon_n$ neighborhood on the Bloch sphere. To approximate all unitary operators, we use circles of radius $\epsilon_n$ to cover the Bloch sphere. Thus we need at least $O(\frac{1}{\epsilon_n^2})$ different circles, or $O(\frac{1}{\epsilon_n^2})$ different operators respectively.
 
A quantum gate array of no more than $L_n$ gates can give at most $2^{O(L_n)}$ different operators. In order to give an interpretation, we consider a universal set of $2^c-1$ gates and a number of $c\cdot L_n$ classical bits. There are $2^{c\cdot L_n}$ possible different numbers we can choose. Each number corresponds to an arrangement of $L_n$ quantum gates: every bit corresponds to a possible operation, for example, we use $0$ to represent that we do nothing and $1$ to $2^c-1$ to represent gates from the universal set. Thus, all possible $c\cdot L_n$-bit numbers corresponds to all possible quantum gate arrays of no more than $L_n$ gates. Different quantum gate arrays may correspond to the same unitary operator, for example, two Hadamard gates are the same as the identity. Thus, an array of at most $L_n$ gates can give at most $2^{O(L_n)}$ different operators.  

Now we can combine the above facts. In order to cover the Bloch sphere, we need $\frac{1}{\epsilon_n^2}$ different operators and thus a quantum gate array of $L_n=O(\log{\frac{1}{\epsilon_n}})$ gates.

\section{Proof of Theorem \ref{thm:approximation}}\label{app:Thm_Approx_Proof}

We will start by presenting the main idea of our proof. Similar to Harrow's proof of the optimality of the Solovay-Kitaev theorem \ref{sec:solovay-kitaev}, we can use segments to cover some points on a section of the Bloch sphere. Midpoints of segments correspond to operators that we can apply accurately, length of segments corresponds to accuracy we need and the points correspond to operators needed in the algorithm for PartialMOD. When a segment cover a point, we can approximate the operator at the point by the operator at the midpoint of the segment. Given the length of segments, we estimate the minimum number of segments, which is also the number of quantum gates that we should apply accurately. Because given some universal set of $N$ gates, we can apply at most $N^L$ different quantum gate arrays with $L$ quantum gates. Thus the logarithm of the minimum number of segments is the length of quantum gate arrays we need. 

In order to prove the theorem, assume for contradiction that such an algorithm exists. We are supposed to approximate all $R(\theta_p)$. Assume that $p$ is sufficiently large and $\epsilon$ is sufficiently small. Assume moreover that we can apply some set of unitary operators accurately, we denote these operators as $R(\theta^i)$. This implies that we can use $R(\theta^i)$ to approximate the quantum gate $R(\theta_p)$, where $\theta_p\in [\theta^i-\epsilon_p, \theta^i+\epsilon_p]$. Here, different $R(\theta_p)$'s can be approximated by the same $R(\theta^i)$ as long as they are in $[\theta^i-\epsilon_p, \theta^i+\epsilon_p]$'s, since
\begin{equation*}
    \theta_p-\theta_{p+1}=\frac {\pi} {2p^2}\ll\epsilon_p=O\left(\frac {\epsilon }{2p}\right)
\end{equation*}
Formally, we denote $R(\theta^i)$ as $R(\frac {\pi}{2q_i})$, and let $\epsilon_i=\frac {\epsilon}{2q_i}$, where $q_i=\frac{\pi}{\theta^i}$ is a suitable real number such that $[\theta^i-\epsilon_p, \theta^i+\epsilon_p]$ covers its neighbouring $\theta_p,\ p\in Z$, and all $\{R(\theta^i)\}$ combined can approximate all unitary operators needed. In order to satisfy this requirement, edges of two nearby segments $R(\theta^i)$ and $R(\theta^{i+1})$ must meet with each other, and therefore the largest distance between them is $\theta^{i}-\theta^{i+1}=\epsilon_{i}+\epsilon_{i+1}$. With this equality we obtain the following recursive formula:
\begin{equation*}
    q_{i+1}=\frac {1+\epsilon}{1-\epsilon}q_i=\left(\frac {1+\epsilon}{1-\epsilon}\right)^i q_1\approx e^{2\epsilon i} q_1
\end{equation*}
Suppose we need to approximate $R(\theta_p),\ \theta_p\in [-\pi,\pi]$, which shows that we have to represent $O(\frac {1}{\epsilon}\log p_0)$ unitary operators accurately. Since a quantum gate array of $n$ gates from a finite universal gate set can give at most $2^{O(n)}$ different unitary operators, as is discussed in \ref{sec:solovay-kitaev}, we must use at least $O(\log (\frac {1}{\epsilon}\log p_0))$ quantum gates. 

One should note that our proof is general for any universal set of quantum gates, because there's no assumption on what universal set we use in the proof. 

\section{Proof of Theorem \ref{thm:another_approximation}}\label{app:Ano_Approx_Proof}

Assume for contradiction that such an algorithm exists, Therefore, we can represent $R(\theta_p)$ to any accuracy of $\epsilon_p$ using up to $p(\log\frac{1}{\epsilon_p})$ quantum gates. We now consider the distance between two adjacent angle
\begin{equation*}
    \theta_p - \theta_{p+1} \approx \frac{\pi}{2p^2}
\end{equation*}
We choose $\epsilon_p=\frac{\pi}{2p^2}$. In this way, we can approximate any operator around $\theta_p$ to an accuracy of $\delta= \frac{2}{\pi}\theta_p^2$ with $o(\log \frac{1}{\delta})$ quantum gates. 

We thus use a similar but simpler recursion with the Solovay-Kitaev algorithm. Suppose we can approximate any operator $U$ by $U_0$ to some basic accuracy $\delta_0$ with an quantum gate array of $L_0$. Find the nearest $R(\theta_{p_1})$ to $R(\delta_0)$. Using the conclusion last paragraph, we can approximate $R(\delta_0)$ to an accuracy of $\delta_1\approx \frac{2}{\pi}\delta_0^2$ with $o(\log\frac{1}{\delta_1})$. After applying $R(\theta_{p_1})$, we approximate $U$ by $R(\theta_{p_1})U_0$ to an accuracy of $\delta_1=\frac{2}{\pi} \delta_0^2$, with  $L_0+o(\log{\frac{1}{\delta_1}})$ quantum gates. After $n$ steps that are similar to the first step, we approximate $U$ by $R(\theta_{p_n})...R(\theta_{p_1})U_0$ to an accuracy of 
\begin{equation*}
    \delta_n=\frac{2}{\pi}\delta_{n-1}^2=(\frac{2}{\pi})^{2^n-1}\delta_0^{2^n}
\end{equation*}
with a quantum gate array of
\begin{equation*}
      L_n=L_{n-1}+o(\log\frac{1}{\delta_n})\approx L_0+\sum_{m=1}^{n} o(\log{\frac{1}{\delta_m}})\approx o(2\log\frac{1}{\delta_n})
\end{equation*}
The above calculation shows that in this way, any gate can be approximate to an accuracy of $\epsilon$ with $o(\log(\frac{1}{\epsilon}))$ quantum gates. This is a contradiction to the optimality of the Solovay-Kitaev algorithm, which states that we can approximate any unitary operator to an accuracy of $\epsilon_n$ with $L_n=O(\log \frac 1 {\epsilon_n})$ quantum gates. 

Therefore, there is no algorithm that can approximate all $R(\theta_p)$ to an accuracy of $\epsilon_p$ with $o(\log \frac 1 {\epsilon_p})$ quantum gates, which concludes the proof. 

\section{Proof of Theorem \ref{thm:EQ_Upper_Bound}}\label{app:EQ_Upper_Bound_Proof}
\begin{proof}
Let $|\delta_{ij}|\leq \delta=\frac 1 n$. Now we choose $k_i,1,...,t=\frac 2 {\epsilon}(n+3)$ uniformly at random from $2^n$ integers of value $0$ to $2^n-1$. Let
\begin{equation*}
    X_l(g)=\sum_{j=1}^{l}{\cos\left(\frac {\pi k_j g}{2^n}\right)};\ 
    Y_l(g)=\sum_{j=1}^{l}\sum_{i}{\sin \left(\frac {\pi k_j 2^i g_i}{2^n}\right)\delta_{ij}}
\end{equation*}
Since $|X_{l+1}(g)-X_l(g)|\leq 1,\ E(X_{l}(g))=0$, and $|Y_{l+1}(g)-Y_l(g)|\leq n\delta,\  E(Y_{l}(g))=0$, the Azuma-Hoeffding theorem \cite{azuma1967} states that for $t\geq \frac {2\ln 2} {\epsilon}(n+3)$ the following two inequalities hold:
\begin{equation*}
    \Pr\left(|X_t(g)|\geq t\sqrt {\epsilon}\right)\leq 2\exp\left(-\frac 1 2 \epsilon t\right)<\frac 1 {2^{n+2}}
\ \ 
\text{and}
\ \ 
\end{equation*}
\noindent and 
\begin{equation*}
    \Pr\left(|Y_t(g)|\geq t\sqrt {\epsilon}\right)\leq 2\exp\left(-\frac {\epsilon t}{2n^2\delta^2}\right)<\frac 1 {2^{n+2}}
\end{equation*}
By combining them, we receive $\Pr\left(|X_t(g)|\geq t\sqrt {\epsilon}\ or \ |Y_t(g)|\geq t\sqrt {\epsilon},\ \forall g\right)< 1 $. This is equivalent to saying 
\begin{equation*}
    \Pr\left(|X_t(g)|\leq t\sqrt {\epsilon}\ and \ |Y_t(g)|\leq t\sqrt {\epsilon},\ \forall g\right)>0
\end{equation*}
Since the possibility that we find a set satisfying our conditions is strictly greater than $0$, there exists a set for which both conditions hold.
\end{proof}
\end{document}